\documentclass[10pt,letterpaper]{IEEEtran}
\usepackage{amsmath,amssymb,amsthm,epsfig,color,subfigure,empheq,graphicx,graphics,balance}
\usepackage{enumerate,url,algorithm,algorithmic,wasysym,epstopdf,enumitem,comment}
\usepackage{xcolor}
\usepackage{amsfonts}
\usepackage{adjustbox}
\usepackage{multirow}

\usepackage{accents}

\newtheorem{assumption}{Assumption}

\newtheorem{lemma}{Lemma}

\newtheorem{theorem}{Theorem}

\newcommand\cmr[1]{\textcolor{red}{#1}}

\newcommand \bzero{\mathbf{0}}
\newcommand \bone{\mathbf{1}}
\newcommand \ba{\mathbf{a}}

\newcommand \bc{\mathbf{c}}
\newcommand \bd{\mathbf{d}}
\newcommand \be{\mathbf{e}}
\newcommand \bef{\mathbf{f}} 
\newcommand \bg{\mathbf{g}}

\newcommand \bi{\mathbf{i}}

\newcommand \bp{\mathbf{p}}
\newcommand \bq{\mathbf{q}}

\newcommand \bs{\mathbf{s}}

\newcommand \bu{\mathbf{u}}
\newcommand \bv{\mathbf{v}}
\newcommand \bw{\mathbf{w}}

\newcommand \bz{\mathbf{z}}

\newcommand \bB{\mathbf{B}}
\newcommand \bC{\mathbf{C}}

\newcommand \bG{\mathbf{G}}

\newcommand \bI{\mathbf{I}}

\newcommand \bS{\mathbf{S}}

\newcommand \bX{\mathbf{X}}
\newcommand \bY{\mathbf{Y}}

\newcommand \btheta{\boldsymbol{\theta}}

\newcommand \mcC{\mathcal{C}}

\newcommand \mcG{\mathcal{G}}

\newcommand \mcL{\mathcal{L}}

\newcommand \mcN{\mathcal{N}}
\newcommand \mcO{\mathcal{O}}
\newcommand \mcP{\mathcal{P}}

\newcommand \mcR{\mathcal{R}}
\newcommand \mcS{\mathcal{S}}

\newcommand \mcV{\mathcal{V}}

\newcommand \bmcS{\bar{\mathcal{S}}}

\newcommand \tba{\tilde{\mathbf{a}}}

\newcommand \tbp{\tilde{\mathbf{p}}}

\newcommand \tbA{\tilde{\mathbf{A}}}
\newcommand \tbB{\tilde{\mathbf{B}}}
\newcommand \tbC{\tilde{\mathbf{C}}}

\newcommand \tbG{\tilde{\mathbf{G}}}

\newcommand \tbM{\tilde{\mathbf{M}}}

\newcommand \tbtheta{\tilde{\boldsymbol{\theta}}}

\newcommand \hbd{\hat{\mathbf{d}}}

\newcommand \hbp{\hat{\mathbf{p}}}

\newcommand \hbv{\hat{\mathbf{v}}}

\newcommand \bbi{\bar{\mathbf{i}}}

\newcommand \bbq{\bar{\mathbf{q}}}

\newcommand \bbv{\bar{\mathbf{v}}}

\newcommand \bbY{\bar{\mathbf{Y}}}

\DeclareMathOperator{\ptdf}{PTDF}

\begin{document}
\title{Expediting AC Contingency Analysis using a Basecase Machine Learning Model}

\author{Md Obaidur Rahman,~\IEEEmembership{Student Member,~IEEE}, and
        Vassilis Kekatos,~\IEEEmembership{Senior Member,~IEEE}	
}
    
\markboth{IEEE TRANSACTIONS ON POWER SYSTEMS (submitted \today)}{IEEE TRANSACTIONS ON POWER SYSTEMS (submitted \today)}

\maketitle

\begin{abstract}
Grid planning and operation under increasingly variable operating conditions require fast and accurate AC power flow (AC-PF) contingency analysis for numerous line outages. While ML models can accelerate these computations, existing approaches often require either training separate models for different contingencies or a single model using data from multiple topologies. Both approaches incur substantial offline data generation and training costs. To address this gap, this work proposes a fixed-point framework that reuses a single ML model trained exclusively on basecase topology data to predict post-contingency AC-PF states under any non-critical single-line outage. We analyze the proposed method's convergence under the DC model approximation and derive its convergence rate in terms of network parameters. As a side result, we show that power transfer distribution factors (PTDFs) for non-critical lines have magnitudes less than one. We further characterize the ML training region to accommodate PF specifications encountered during fixed-point iterations and derive bounds on ML prediction errors for post-contingency state estimates. Numerical tests on a 6,717-bus synthetic Texas system demonstrate convergence across all non-critical single-line outages, with prediction errors remaining close to those of basecase. The proposed framework offers a favorable tradeoff between the speed of DC solvers and Newton-Raphson accuracy.\end{abstract}


\section{Introduction}\label{sec:intro}
\allowdisplaybreaks
Contingency analysis is a fundamental yet computationally formidable task in power system operation and planning. It assesses the system's ability to withstand credible component outages that may lead to line overloads, voltage instability, or cascading failures. A power system is deemed $N-k$ secure if it can withstand any credible failure of $k$ major components, such as transmission lines, transformers, or generators~\cite{NERC2010}. In its preventive form, contingency analysis ensures that the system remains secure under any contingency without requiring any adjustments. In its corrective form, the analysis provides an action plan for each contingency, assuming the operator has sufficient time to respond. The rapid integration of data center loads and inverter-interfaced distributed energy resources intensifies the need for increasingly refined AC contingency analyses. This work shows how a system operator can repurpose a single ML model trained to predict basecase AC-PF states to compute $N-1$ post-contingency AC-PF states.  

To expedite contingency analysis, substantial work has studied sensitivity-based approximations of post-contingency power flows. Power transfer distribution factors (PTDFs) and line outage distribution factors (LODFs) provide closed-form estimates of active power flows under the DC-PF model~\cite{wood2013power}. Subsequent studies developed more efficient procedures for evaluating these factors. For example, \cite{guler2007generalized}, and \cite{castelli2024improved} derived direct LODF formulations and methods that exploit the sparsity of grid topologies. Reference~\cite{ZhGi12} leverages the sparse representation of line outages to detect topology changes. These methods are well suited for rapid contingency screening, but their reliance on the DC-PF approximation precludes an accurate representation of voltage magnitudes, network losses, and apparent power limits.

To estimate these quantities while preserving computational efficiency, researchers have developed linearized approximations of the AC-PF equations. Researchers have used AC sensitivity factors for security assessment under uncertain operating conditions~\cite{qu2015uncertainty}. Reference~\cite{dhople2015linear} derived linear approximations of the AC-PF equations in rectangular coordinates, and \cite{coffrin2014linear} developed a linear-programming approximation of AC power flow. Although such approaches are more accurate than the DC-PF model, their performance depends heavily on the linearization point and the reference topology. To avoid linearization inaccuracies, reference~\cite{yu2024efficient} suggested warm-starting the post-contingency AC-PF solver from the pre-contingency state. In fact, the inverse PF Jacobian required in the first Newton-Raphson (NR) iteration for the post-contingency state can be efficiently computed from the pre-contingency inverse PF Jacobian~\cite{yeung2017amps}. Holomorphic embedding has also been applied to AC contingency analysis to improve convergence near stressed operating conditions where NR methods may fail to converge~\cite{yao2021contingency}. Although these methods improve the efficiency or robustness of post-contingency AC-PF solvers, contingency analysis still requires repeated solutions across many loading and outage scenarios. This computational burden has motivated data-driven surrogate models.


Among these data-driven models, DNN-based approaches have been developed for learning optimal power flow (OPF) solutions~\cite {zamzam2020learning,pan2021deepopf,singh2022learning}, corrective actions~\cite{zhou2022scalable}, state estimation and in other grid design and security assessment tasks \cite{anna,gupta2023optimal,chen2020learning,kody2023learning}. Physics-guided NNs incorporate the AC-PF equations into the training process to improve the physical consistency of the predicted states~\cite{hu2021physics,eeckhout2024improved}. Learning-based methods have also been suggested specifically for contingency-constrained problems. Reference~\cite{du2019achieving} accelerates $N-1$ contingency screening using a deep convolutional neural network to find post-contingency AC-PF solutions. However, it requires training on numerous scenarios and topologies.

To reduce dependence on topology-specific training, researchers have explored topology-aware learning architectures. Graph neural networks (GNNs) address this issue by incorporating the power system topology directly into the learning architecture and have been applied to contingency prediction \cite{nakiganda2025graph}. More recently, \cite{wu2026universal} developed a universal GNN for power system state estimation across different topologies. Nonetheless, training across multiple scenarios and topologies remains a major computational burden. Reference~\cite{christianson2025fast} uses an input-convex NN to classify whether a power system is $N-k$ secure under the DC-PF model, yet that is a relatively lighter task than predicting the post-contingency states. 

Can a model trained solely to solve the AC-PF problem under the basecase topology be directly reused to predict post-contingency states? This work answers this question affirmatively and makes four contributions:
\begin{itemize}
\item[\emph{c1)}] We propose predicting post-contingency AC-PF states via a fixed-point iteration that reuses a single ML model trained exclusively on the basecase topology (Section~\ref{sec:problem}). 

\item[\emph{c2)}] To study the convergence of the proposed scheme, we analyze the fixed-point iterations derived from the DC-PF model and establish convergence for all non-critical single-line outages (Section~\ref{sec:convergence}). As a byproduct of broader interest, we establish bounds on PTDFs. 

\item[\emph{c3)}] We characterize the expanded set of AC-PF specifications needed to train the base-case ML model so that the fixed-point iteration operates properly (Section~\ref{sec:region}).

\item[\emph{c4)}] We analyze how ML prediction errors propagate through the fixed-point iteration to the final post-contingency state estimate (Section~\ref{sec:accuracy}).
\end{itemize}
We numerically tested the proposed method on the synthetic Texas 6,717-bus system using a DNN trained solely using basecase data. The proposed fixed-point iteration converges for all non-critical line contingencies. The empirical convergence rate and prediction accuracy corroborate our analytical findings. Overall, the method achieves a favorable speed–accuracy tradeoff, substantially reducing runtime relative to NR while maintaining higher accuracy over the DC-based solver. 

The questions addressed under \emph{c2)} and \emph{c3)} were also studied in the conference precursor of this work~\cite{HICSS2027}. In~\cite{HICSS2027}, we relied on semidefinite programming (SDP) formulations to bound the sensitivity of the basecase DNN with respect to its AC-PF specification inputs. Although this SDP-based analysis bounds the exact DNN fixed-point iteration, it is computationally demanding, limited to DNNs, and can yield loose sensitivity bounds that make the resulting guarantees overly conservative. In contrast, the DC-PF-based analysis developed here applies to arbitrary ML models, does not depend on the pretrained model, and provides insight into which contingencies may require more iterations or a larger training set. 

\section{Problem Statement and Formulation}\label{sec:problem}
This section defines the AC contingency analysis problem and proposes solving it with a single ML model trained on the basecase topology. A power system can be represented by a graph $\mcG_0 = (\mcN, \mcL_0)$. The set of nodes $\mcN$ collects the $N+1$ buses indexed by $n=0, 1, \dots, N$. The set of edges $\mcL_0:=\{(m,n):m,n\in\mcN\}$ collects $L$ transmission lines energized under the basecase system topology. The lines are indexed by $\ell=1,\ldots,L$. We partition $\mcN$ into three mutually exclusive sets: the singleton of the reference bus indexed by $n=0$; the set $\mcV$ of generator buses; and the set of load buses $\mcP$. Zero-injection buses belong to set $\mcP$. Buses hosting both generation and loads belong to set $\mcV$. For each bus $n\in\mcN$, let $v_n = V_n \angle \theta_n$ denote the voltage phasor, $i_n$ the injection current phasor, and $s_n = p_n + jq_n$  the complex power injection into bus $n$. Let vectors $\bv$, $\bi$, and $\bs$ stack the bus voltages, currents, and complex power injections for all $n\in\mcN$, respectively. 

Under the basecase topology, the AC-PF equations can be compactly expressed as
\begin{equation} \label{eq:pf0}
   \bs = \bv \odot \bbi = \bv \odot \left(\bbY_0\bbv\right), 
\end{equation}
where $\bY_0$ is the bus admittance matrix, $\odot$ denotes entrywise multiplication, and barred symbols denote complex conjugation. Matrix $\bY_0$ can be expressed as
\begin{equation}\label{eq:Y0}
\bY_0 = \sum_{\ell\in\mcL_0} \left(
y_\ell \ba_\ell\ba_\ell^\top + 
y_\ell^s\left(\be_m\be_m^\top+\be_n\be_n^\top\right) \right),
\end{equation}
where $y_\ell$ and $y_\ell^s$ are the series admittance and half shunt admittance of line $\ell$, respectively. If line $\ell=(m,n)$ runs between buses $m$ and $n$, we define vector
\begin{equation}\label{eq:al}
\ba_\ell:=\be_m-\be_n\quad \forall \ell\in\mcL_0,
\end{equation}
where $\be_m$ is the $m$-th column of the identity matrix $\bI_{N+1}$. 

The AC-PF problem is to find a power system state vector $\bv$ that satisfies \eqref{eq:pf0}, given two specifications per bus. The specification types vary with the bus type: \emph{i)} For the reference bus, we specify $V_0$ and set $\theta_0=0$; \emph{ii)} For PV buses, we specify $(p_n,V_n)$; and \emph{iii)} For PQ buses, we specify $(p_n,q_n)$. Given these $(2N+1)$ real-valued specifications, we can find an AC-PF solution by solving \eqref{eq:pf0} using the Newton-Raphson (NR) or other numerical methods. 

Assuming a particular AC-PF solver, define the mapping $G_0:\mathbb{R}^{2N+1}\rightarrow \mathbb{C}^{N+1}$ from the vector of AC-PF specifications $\bc\in\mathbb{R}^{2N+1}$ to the associated AC-PF solution
\begin{equation}\label{eq:pfmapping}
\bv_0 := G_0(\bc).
\end{equation}
The subscript $0$ emphasizes that this mapping refers to the basecase topology $\mcL_0$. 

In addition to basecase, a system operator is interested in finding the power system state under any single-line contingency. Let $\ell$ index the outage of line $\ell\in\mcL_0$. The system topology under contingency $\ell$ is defined over the edge set
\[\mcL_\ell:=\mcL_0\setminus \{\ell\}.\]
After contingency $\ell$, the power system graph becomes $\mcG_\ell = (\mcN, \mcL_\ell)$. If graph $\mcG_\ell$ is disconnected, we call line $\ell$ a \emph{critical line}. We henceforth assume that contingency analysis does not study the outage of critical lines. To simplify the exposition, let us assume that all lines are non-critical.

As with the basecase topology, the operator can use an AC-PF solver to find the system state $\bv_\ell$ under contingency $\ell$. Similar to $G_0$, the related mapping $G_\ell:\mathbb{R}^{2N+1}\rightarrow \mathbb{C}^{N+1}$ from AC-PF specifications to the post-contingency state is
\begin{equation}\label{eq:pfmappingell}
\bv_\ell = G_\ell(\bc).
\end{equation}
Clearly, running contingency analysis under the AC-PF model entails solving $(L+1)$ AC-PF problems. Under preventive contingency analysis, all contingencies share the same AC-PF specification vector $\bc$. Under corrective contingency analysis, the AC-PF specifications may vary per contingency. Our proposed methodology applies to both settings. 

To expedite AC contingency analysis, we first relate each $G_\ell$ to $G_0$. The bus admittance matrix $\bY_\ell$ associated with contingent topology $\mcL_\ell$ can be expressed as
\begin{equation}\label{eq:Y_ell}
\bY_\ell=\bY_0-y_\ell\ba_\ell\ba_\ell^\top -y_\ell^s\left(\be_m \be_m^\top+\be_n \be_n^\top\right).
\end{equation}
The contingency state satisfies the modified AC-PF equations
\begin{equation}\label{eq:pfell0}
    \bs = \bv_{\ell} \odot \left(\bbY_\ell\bbv_\ell\right).
\end{equation}
Substituting \eqref{eq:Y_ell} into \eqref{eq:pfell0} and rearranging yields
\begin{equation}\label{eq:pfell}
\bs +\check{\bg}_\ell(\bv_\ell)= \bv_{\ell} \odot \left(\bbY_0 \bbv_{\ell} \right),
\end{equation}
where the mapping $\check{\bg}_\ell(\bv_\ell):\mathbb{C}^{N+1}\rightarrow\mathbb{C}^{N+1}$ is defined as
\begin{equation}\label{eq:gell}
\check{\bg}_\ell(\bv_\ell):=\bv_\ell\odot\color{black}\left(\bar{y}_\ell\ba_\ell\ba_\ell^\top\bbv_\ell 
+ \bar{y}_\ell^s\left(\be_m \be_m^\top+\be_n \be_n^\top\right)\bbv_\ell\right).
\end{equation}
Vector $\check{\bg}_\ell(\bv_\ell)$ has only two non-zero entries, corresponding to the terminal buses $(m,n)$ of the outaged line $\ell$. 

Interestingly, the contingent PF equations in \eqref{eq:pfell} can be interpreted as the basecase PF equations in \eqref{eq:pf0} if complex power injections are modified from $\bs$ to $\bs +\check{\bg}_\ell(\bv_\ell)$. A similar concept, oftentimes termed the \emph{compensation trick}, appears in contingency analysis under the linearized DC-PF model~\cite[Sec.\,7.4]{wood2013power}. Under the DC-PF model, the compensation trick applies to a system of linear equations. In contrast, the compensation trick here applies to a system of quadratic equations.

Thanks to this interpretation, the system state under contingency $\ell$ satisfies
\begin{equation}\label{eq:equilibrium}
\bv_{\ell} = G_0\left(\bc + \bg_{\ell}(\bv_{\ell})\right),
\end{equation}
where the mapping $\bg_{\ell}(\bv_{\ell}):\mathbb{C}^{N+1}\rightarrow \mathbb{R}^{2N+1}$ is defined based on $\check{\bg}_\ell(\bv_\ell)$ as follows: \emph{i)} For the entries of $\bc$ corresponding to active power injections into PQ/PV buses, the mapping $\bg_{\ell}(\bv_{\ell})$ extracts the real part of $\check{\bg}_{\ell}(\bv_{\ell})$; \emph{ii)} For the entries of $\bc$ corresponding to reactive power injections into PV buses, the mapping $\bg_{\ell}(\bv_{\ell})$ extracts the imaginary part of $\check{\bg}_{\ell}(\bv_{\ell})$; and \emph{iii)} For the entries of $\bc$ corresponding to voltage magnitudes at the reference and the PV buses, the mapping $\bg_{\ell}(\bv_{\ell})$ takes zero values. Therefore, depending on whether the terminal buses of line $\ell$ are PQ, PV, or the reference bus, the vector $\bg_\ell(\bv_\ell)$ has 2--4 nonzero entries. 

Equation~\eqref{eq:equilibrium} establishes that any post-contingency state can be defined with the help of the basecase AC-PF mapping. Unfortunately, finding $\bv_\ell$ from \eqref{eq:equilibrium} is not straightforward, as the state appears on either side of this equation. Nonetheless, one may attempt finding $\bv_\ell$ using the fixed-point iteration
\begin{equation}\label{eq:pf-fixed-point}
\bv_{\ell}^{t+1} = G_0\left(\bc + \bg_{\ell}(\bv_{\ell}^t)\right).
\end{equation}
Each iteration $t$ involves two steps: In the first step, we modify 2--4 entries of $\bc$ by evaluating $\bg_\ell$ at the previous state estimate $\bv_\ell^t$. In the second step, we feed the modified specifications to the basecase AC-PF solver and receive the updated state $\bv_\ell^{t+1}$. If this iteration converges, its limit point satisfies~\eqref{eq:equilibrium}. However, finding $\bv_\ell$ using \eqref{eq:pf-fixed-point} rather than \eqref{eq:pfmappingell} bears no practical advantage. The fixed-point iteration invokes the AC-PF solver $G_0$ multiple times, which is clearly computationally more expensive than invoking $G_\ell$ just once. The significance of \eqref{eq:pf-fixed-point} lies in the insight it provides for translating this fixed-point iteration idea to an ML setting, as described next.

To accelerate AC-PF studies, a system operator may use an ML model (such as a DNN, GNN, or Gaussian process) to learn the basecase AC-PF mapping. Suppose we have already trained an ML model to predict AC-PF solutions given a specification vector $\bc$ as input. If vector $\bw$ collects the ML tunable parameters, we denote the ML mapping as
\begin{equation}
\bv = F_\bw(\bc).
\end{equation}

How about the contingency topologies? Given the large number of transmission lines, training a different ML model for each contingent topology is impractical. The fixed-point iteration in \eqref{eq:pf-fixed-point} suggests a possibly viable alternative. It suffices to train a \emph{single} ML model under the basecase topology to predict the state under any contingency $\ell$ by iteratively invoking the basecase ML model as
\begin{equation}\label{eq:dnn-fixed-point}
\bv_{\ell}^{t+1} = F_\bw\left(\bc + \bg_{\ell}(\bv_{\ell}^t)\right).
\end{equation}
The proposed approach incorporates the basecase ML model into a fixed-point iteration, as shown in Fig.~\ref{fig:loop}. 

\begin{figure}[t]
\centering
\includegraphics[width=1\linewidth]{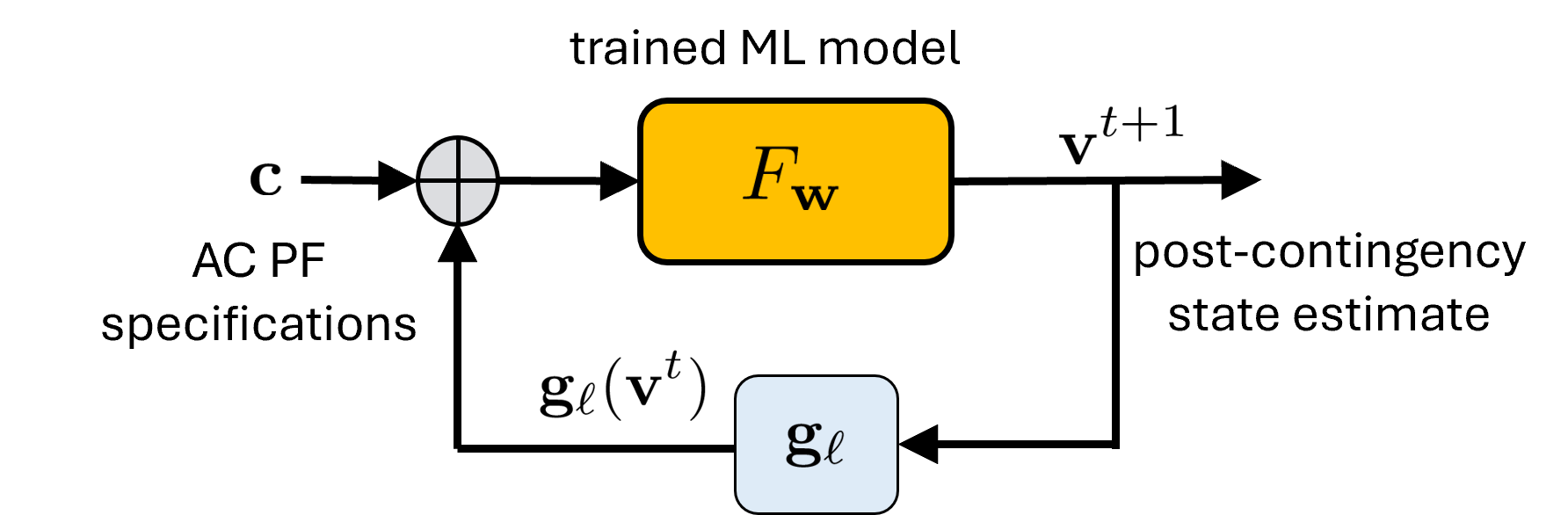}
\caption{Fixed-point formulation for predicting a post-contingency state using an ML model trained solely on the basecase topology. The feedback loop depends on the contingency. For a single-line contingency, the mapping $\mathbf{g}_{\ell}(\mathbf{v}^{t})$ is nonzero only at the terminal buses of the outaged line.}
\label{fig:loop}
\vspace*{-1em}
\end{figure}

During training, the proposed approach is meaningful because it avoids training $L$ additional ML models. During inference, the proposed approach can be computationally attractive if the fixed-point iterations converge quickly and the ML inference step is faster than an AC-PF solver. Before evaluating the proposed methodology numerically in Sec.~\ref{sec:tests}, we address the following three questions:
\begin{enumerate}
    \item[\emph{q1)}] Does the fixed-point iteration in \eqref{eq:dnn-fixed-point} converge for all contingencies? If so, what is its convergence rate? This question is addressed in Sec.~\ref{sec:convergence}.
    \item[\emph{q2)}] How can we ensure that the ML model has been trained to predict states accurately when given $\bc + \bg_{\ell}(\bv_{\ell}^t)$ as its input for all contingencies $\ell$, specification vectors $\bc$, and iterations $t$? This question is addressed in Sec.~\ref{sec:region}.
    \item[\emph{q3)}] How does the prediction error of the ML model propagate throughout iterations into the final post-contingency state estimate? This question is addressed in Sec.~\ref{sec:accuracy}.
\end{enumerate}

\section{Convergence Rate Analysis}\label{sec:convergence}
Establishing the convergence of \eqref{eq:dnn-fixed-point} is challenging because this fixed-point iteration involves an ML model coupled with the nonlinear mapping $\bg_\ell(\bv)$. A natural approach to ensure \eqref{eq:dnn-fixed-point} reaches an equilibrium is to show that $F_\bw(\bc+\bg_\ell(\bv))$ is a contraction mapping with respect to $\bv$, upon invoking Banach's fixed-point theorem~\cite{banachcontraction}. To do so, we need to characterize the smoothness of the ML model mapping through its Lipschitz constant. In~\cite{HICSS2027}, we bounded this constant by solving an SDP. Nevertheless, this SDP formulation applies only to DNNs, may be computationally expensive, and may yield overly conservative upper bounds. 

Alternatively, if the ML model has learned the AC-PF mapping accurately enough, one may try establishing the convergence of \eqref{eq:pf-fixed-point} rather than \eqref{eq:dnn-fixed-point}. 
However, proving $G_0\left(\bc + \bg_{\ell}(\bv_{\ell})\right)$ is contracting is also technically challenging. Considering the AC-PF solver alone, existing convergence results rely on contractions, energy functions, or monotone operators, but apply only to lossless networks, radial networks, or limited loading conditions~\cite{simpson2017theory,bazrafshan2018convergence, bolognani2016}.

Given this predicament, we approximate the AC-PF mapping in \eqref{eq:pf-fixed-point} by the linearized mappings corresponding to the DC-PF model. For the basecase, the DC-PF mappings are
\begin{subequations}\label{eq:dcpf}
\begin{align}
\bp&=\bB_0\btheta,\label{eq:ptheta}\\
\bq&=\bB_0\bu,\label{eq:qv}
\end{align}
\end{subequations}
where vectors $(\bp,\bq)$ collect the active and reactive power injections at all buses, vector $\btheta$ collects the bus voltage angles, vector $\bu$ collects the bus voltage magnitude deviations from one per unit, and the bus susceptance matrix is defined as
\begin{equation}\label{eq:B0-dc}
\bB_0:=\sum_{k\in\mcL_0}\frac{1}{x_k}\ba_k\ba_k^\top,
\end{equation}
where $x_k>0$ is the reactance of line $k$. To simplify the exposition and without loss of generality, we use the same matrix $\bB_0$ for the two components of the DC-PF model. 

Despite its approximate nature, the convergence analysis relying on the DC-PF model reveals the key factors governing the convergence of \eqref{eq:pf-fixed-point} and \eqref{eq:dnn-fixed-point}. Recall that \eqref{eq:pf-fixed-point} maps AC-PF specifications to complex voltages. Under the DC-PF approximation, the exact AC-PF mapping $G_0$ decouples into the $P{-}\theta$ and the $Q{-}V$ mappings in \eqref{eq:dcpf}. We next study the convergence of these two mappings. 

\vspace*{-.5em}
\subsection{Convergence of the $P{-}\theta$ Mapping}
\label{subsec:ptheta}
If we drop the first row and column from $\bB_0$, we obtain the reduced bus susceptance matrix defined as
\begin{equation}
\tbB_0 := \sum_{k\in\mcL_0} \frac{1}{x_k}\tba_k\tba_k^\top,
\end{equation}
where vector $\tba_k$ is obtained from $\ba_k$ upon dropping its first entry corresponding to the reference bus. Similar to the AC case, the matrix appearing in the DC-PF model under contingency $\ell$ is a rank-one modification of $\tbB_0$. Hence, the reduced $P$--$\theta$ model under contingency $\ell$ reads as
\begin{equation}\label{eq:DC_PF_Contingency}
\tbp=\left(\tbB_0-\frac{1}{x_\ell}\tba_\ell\tba_\ell^\top
\right)\tbtheta_\ell,
\end{equation}
where $\tbp$ and $\tbtheta$ are $N$-length vectors excluding  the reference bus. Rearranging terms yields
\begin{equation}\label{eq:dcpf-equilibrium}
\tbtheta_\ell=\tbG_\ell\tbtheta_\ell+\tbB_0^{-1}\tbp,
\end{equation}
where $\tbG_\ell:=\frac{1}{x_\ell}\tbB_0^{-1}\tba_\ell\tba_\ell^\top$. Equation~\eqref{eq:dcpf-equilibrium} naturally gives rise to the fixed-point iteration
\begin{equation}\label{eq:dcpf-fixed-point}
\tbtheta_\ell^{t+1}=\tbG_\ell\tbtheta_\ell^t+\tbB_0^{-1}\tbp,
\end{equation}
which is a discrete-time affine system. Before establishing its stability, we review power transfer distribution factors (PTDFs) and derive a key property instrumental to subsequent analysis.

According to the DC-PF model, the vector $\bef\in\mathbb{R}^L$ of active power line flows depends linearly on active power injections
\begin{equation}\label{eq:f=Sp}
\bef=\bS\bp,
\end{equation}
where $\bS:=\bX^{-1}\tbA\tbB_0^{-1}$, the diagonal matrix $\bX$ carries the line reactances on its main diagonal, and $\tbA$ has vectors $\{\tba_k\}_{k=1}^L$ as its rows. The $(\ell,m)$-th entry of $\bS$, denoted by $S_{\ell,m}$, quantifies the sensitivity of line flow $\ell$ to bus injection $m$. If one unit of power is transferred from bus $i$ to bus $j$, the resulting change in line flow $\ell$ is given by the \emph{power transfer distribution factor}
\begin{equation}\label{eq:ptdf}
\ptdf_{\ell,i,j}:= S_{\ell,i}-S_{\ell,j}.
\end{equation}
If bus $i$ or $j$ is the reference bus, set $S_{\ell,i}=0$ or $S_{\ell,j}=0$, accordingly. 
We next show a key property of PTDFs.

\begin{theorem}\label{th:ptdf}
For any non-critical line $\ell$, it holds that
\begin{equation*}
|\ptdf_{\ell,i,j}|<1,
\qquad\forall\,i,j\in\mcN.
\end{equation*}
\end{theorem}

Although this property is widely used, we could not find a formal proof in the literature. For completeness, a proof of Theorem~\ref{th:ptdf} is provided in the appendix. We next use Th.~\ref{th:ptdf} to establish the convergence of \eqref{eq:dcpf-fixed-point}. While Th.~\ref{th:ptdf} applies to any PTDF triplet $(\ell,i,j)$, the next result requires the property only for the triplet $(\ell,m,n)$, where $m$ and $n$ are the terminal buses of line $\ell$. We will refer to $\ptdf_{\ell,m,n}$ as the \emph{self-PTDF} of line $\ell=(m,n)$, and denote it with the more compact notation
\begin{equation}\label{eq:selfptdf}
\beta_\ell:=\ptdf_{\ell,m,n}.
\end{equation}

\begin{theorem}\label{th:pthetacontraction}
The $P{-}\theta$ fixed-point iteration in \eqref{eq:dcpf-fixed-point} converges to the solution of \eqref{eq:dcpf-equilibrium} for any non-critical line $\ell=(m,n)\in\mcL$. More specifically, for every $t\geq 1$, it holds that
\begin{equation}\label{eq:rate}
\|\btheta_\ell^t-\btheta_\ell\|_2
\le
\beta_\ell^{t-1}
\cdot
\|\tbG_\ell(\btheta_\ell^0-\btheta_\ell)\|_2,
\end{equation}
provided that the initial point $\btheta_\ell^0$ satisfies
$\tba_\ell^\top(\btheta_\ell^0-\btheta_\ell)\neq 0$. The self-PTDF satisfies $\beta_\ell\in(0,1)$.
\end{theorem}


Theorem~\ref{th:pthetacontraction} establishes that the fixed-point iteration in \eqref{eq:dcpf-fixed-point} converges geometrically for all non-critical lines. It also shows that contingencies on lines with self-PTDF values near one may require more iterations to converge. The condition on the initial point is mild and is satisfied, for example, by $\btheta_\ell^0=\bzero$. To complete the approximate convergence analysis and address question \emph{q1)}, we establish analogous convergence results for the $Q{-}V$ DC-PF fixed-point iteration. Interested readers can find this analysis in the next subsection; otherwise, we proceed to question \emph{q2)} in Sec.~\ref{sec:region}.

\vspace*{-.5em}
\subsection{Convergence of the $Q{-}V$ Mapping}\label{subsec:qv}


Despite the resemblance between \eqref{eq:ptheta} and \eqref{eq:qv}, studying the convergence of the $Q{-}V$ DC-PF fixed-point iteration is slightly more involved. Partition \eqref{eq:qv} into PQ/PV buses as
\begin{equation}\label{eq:qvpart}
\begin{bmatrix}
\bq_\mcP\\
\bq_\mcV
\end{bmatrix}
=
\begin{bmatrix}
\bB_{\mcP\mcP} & \bB_{\mcP\mcV}\\
\bB_{\mcV\mcP} & \bB_{\mcV\mcV}
\end{bmatrix}
\begin{bmatrix}
\bu_\mcP\\
\bu_\mcV
\end{bmatrix}.
\end{equation}
For this analysis only, we treat the reference bus as a PV bus. This is because its voltage magnitude is specified and active power injections are immaterial in the $Q{-}V$ model.

Different from the $P{-}\theta$ model, here we are given $(\bq_\mcP,\bu_\mcV)$ and want to find $(\bq_\mcV,\bu_\mcP)$. In fact, the subvector $\bu_\mcP$ is the key unknown in \eqref{eq:qvpart}. Once we know $\bu_\mcP$, we can find $\bq_\mcV$ as
$\bq_\mcV=\bB_{\mcV\mcP}\bu_\mcP + \bB_{\mcV\mcV}\bu_\mcV$. 
Hence, we focus on the top row partition of \eqref{eq:qvpart} and rearrange it as
\begin{equation}\label{eq:qvpart2}
\bB_{\mcP\mcP}\bu_\mcP = \bq_\mcP - \bB_{\mcP\mcV}\bu_\mcV.
\end{equation}

Equation~\eqref{eq:qvpart2} applies for the basecase. Under contingency $\ell$, the non-reduced bus susceptance matrix is modified as 
\begin{equation}
\bB_\ell = \bB_0-\frac{1}{x_\ell} \ba_\ell\ba_\ell^\top.
\label{eq:basecase_qv}
\end{equation}
Let us partition $\ba_\ell$ conformably with \eqref{eq:qvpart} as
\[\ba_\ell = \begin{bmatrix}
\ba_{\ell_\mcP}\\
\ba_{\ell_\mcV}
\end{bmatrix}.\]
Per \eqref{eq:qvpart2}, the $Q{-}V$ model under contingency $\ell$ reads as
\begin{equation}\label{eq:dcqv-equilibrium}
\left(\bB_{\mcP\mcP}-\frac{1}{x_\ell} \ba_{\ell_\mcP}\ba_{\ell_\mcP}^\top\right)\bu_\mcP = \underbrace{\bq_\mcP - \left(\bB_{\mcP\mcV}-\frac{1}{x_\ell} \ba_{\ell_\mcP}\ba_{\ell_\mcV}^\top\right)\bu_\mcV}_{:=\bbq_{\mcP}^\ell}.
\end{equation}
The vector on the right-hand side can be computed beforehand because both $\bq_\mcP$ and $\bu_\mcV$ are known. We can now convert \eqref{eq:dcqv-equilibrium} into a fixed-point iteration to find the unknown voltage magnitudes at PQ buses under contingency $\ell$ as
\begin{equation}\label{eq:dcqv-fixed-point}
\bu_\mcP^{t+1}=\bG_\ell\bu_\mcP^{t} + \bB_{\mcP\mcP}^{-1}\bbq_{\mcP}^\ell
\end{equation}
where $\bG_\ell:=\frac{1}{x_\ell}\bB_{\mcP\mcP}^{-1}\ba_{\ell_\mcP}\ba_{\ell_\mcP}^\top$. The matrix $\bB_{\mcP\mcP}$ is invertible because it is a principal submatrix of $\tbB_0$.
Matrix $\bG_\ell$ has rank at most one. Its strictly positive eigenvalue equals to its trace
\begin{equation}\label{eq:lambda-qv}
\lambda_\ell=
\frac{1}{x_\ell}\ba_{\ell_{\mcP}}^\top\bB_{\mcP\mcP}^{-1}
\ba_{\ell_{\mcP}}.
\end{equation}
We next establish that \eqref{eq:dcqv-fixed-point} converges for all non-critical lines.

\begin{theorem}\label{th:qvcontraction}
The $Q{-}V$ fixed-point iteration in \eqref{eq:dcqv-fixed-point} converges to the solution of \eqref{eq:dcqv-equilibrium} for any non-critical line $\ell=(m,n)\in\mcL_0$. In particular, the eigenvalue of matrix $\bG_\ell$ in \eqref{eq:dcqv-fixed-point} satisfies $0\leq\lambda_\ell<1$.
Consequently, for every $t\geq1$,
\begin{equation}\label{eq:qv-rate}
\|\bu_{\ell,\mcP}^t-\bu_{\ell,\mcP}\|_2\leq\lambda_\ell^{t-1}
\|\bG_\ell(\bu_{\ell,\mcP}^0-\bu_{\ell,\mcP})\|_2.
\end{equation}
\end{theorem}

\begin{figure}[t]
\centering
\includegraphics[width=.9\linewidth]{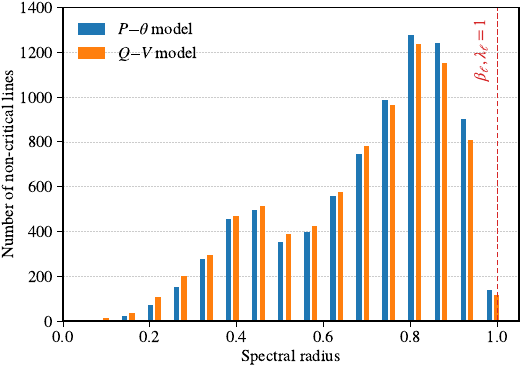}
\vspace*{-1em}
\caption{Distribution of contraction constants involved in the $P{-}\theta$ and the $Q{-}V$ models across all 8,079 non-critical lines of the Texas7k system.}
\label{fig:spectral_radius}
\end{figure}

We finish this section by illustrating the distribution of the contraction constants $\beta_\ell$ involved in the $P{-}\theta$ model (Th.~\ref{th:pthetacontraction}) and $\lambda_\ell$ in the $Q{-}V$ model (Th.~\ref{th:qvcontraction}) for the Texas7k benchmark system. This is a synthesized 6,717-bus system with 8,079 non-critical lines out of 9,140 transmission lines. Figure~\ref{fig:spectral_radius} shows the distributions of these spectral radii across the non-critical lines, confirming they are all less than one.


\color{black}
\section{Admissible Region for ML Model Inputs}\label{sec:region}
Suppose an operator wants to study a power system for all specification vectors $\bc$ in a given set $\mcC\subset \mathbb{R}^{2N+1}$. If the operator is interested only in the basecase topology, they can train an ML model to predict basecase AC-PF solutions for sufficiently many specification scenarios drawn from $\mcC$. In this case, the \emph{training region} of the ML model coincides with the \emph{specification region} of the AC-PF study. However, if the operator intends to use the same model as in \eqref{eq:dnn-fixed-point} to conduct contingency analysis over the specification region $\mcC$, the model must be trained over an expanded training region $\mcC'$, with $\mcC\in\mcC'$. This section explains how to find $\mcC'$ given $\mcC$.

\emph{Why should the ML model be trained on an expanded training set?} For the fixed-point iteration of \eqref{eq:dnn-fixed-point} to succeed, the ML model should accurately predict the AC-PF solutions for all modified specification vectors
\begin{equation}\label{eq:modifiedspec}
\bc_\ell^t:=\bc+\bg_\ell(\bv_\ell^t).
\end{equation}
This should hold for all $\bc\in\mcC$, all contingencies $\ell\in\mcL$, and across iterations indexed by $t$. Therefore, the training set of the ML model $\mcC'$ should include all possible values of $\bc_\ell^t$ from \eqref{eq:modifiedspec}. If we set $\bv_\ell^t=\bone$ (flat voltage profile) in \eqref{eq:modifiedspec}, we get $\bc_\ell^t=\bc$. This implies that if $\bc\in\mcC$, then $\bc\in\mcC'$, and thus, $\mcC\subset\mcC'$. The cardinal question is how to find $\mcC'$.

Determining $\mcC'$ to include all vectors $\bc_\ell^t$ is a formidable task because \eqref{eq:dnn-fixed-point} is nonlinear, iterative, and involves the ML model. To keep the analysis tractable, we again resort to the DC-PF models of Sec.~\ref{sec:convergence}. Moreover, we postulate that the specification set admits the Cartesian-product structure
\begin{equation}\label{eq:cartesian}
\mcC=\mcC_p\times\mcC_q\times\mcC_u,
\end{equation}
where $\mcC_p$, $\mcC_q$, and $\mcC_u$ are the sets of active/reactive power and voltage magnitude specifications, respectively. 

Due to the decoupled nature of the DC-PF models, the training set $\mcC'$ inherits the same Cartesian-product structure
\begin{equation}\label{eq:cartesian2}
\mcC'=\mcC_p'\times\mcC_q'\times\mcC_u'
.
\end{equation}
Because $\bg_\ell(\bv)$ modifies only power injections and no voltage magnitudes, we can safely set $\mcC_u'=\mcC_u$. In the subsequent analysis, we rely on the $P{-}\theta$ model to approximate $\mcC_p'$ from $\mcC_p$, and on the $Q{-}V$ model to approximate $\mcC_q'$ from $\mcC_q$. 

Using the $P{-}\theta$ model, we first derive a sufficient condition for $\mcC_p'$ to be a valid training set. 

\begin{lemma}\label{le:Cp'}
According to the $P{-}\theta$ model, if the training set $\mcC_p'$ includes the points 
\[\mu\tbp+(1-\mu)\left(\bI_N+\frac{1}{1-\beta_\ell}\tbG_\ell^\top\right)\tbp\]
for all $\mu\in[0,1]$, $\tbp\in\mcC_p$, and contingencies $\ell$, then $\mcC_p\subset \mcC_p'$.
\end{lemma}

Based on Lemma~\ref{le:Cp'}, we next provide an outer approximation of $\mcC_p'$ for the case that $\mcC_p$ takes a polytopic form as
\begin{equation}\label{eq:Cp-polytope}
\mcC_p:=\left\{\tbp:\bC_p\tbp\leq\bd_p\right\}.
\end{equation}
where $\bC_p\in\mathbb{R}^{I_p\times N}$. This polytopic form is reasonable, since $\mcC_p$ can be specified by imposing limits on active injections and line flows, according to \eqref{eq:f=Sp}.

According to Lemma~\ref{le:Cp'}, the set $\mcC_p'$ is the union of $L$ polytopes. Because computing $\mcC_p'$ exactly is difficult, we resort to an outer convex approximation $\hat{\mcC}_p$ of $\mcC_p'$, so that $\mcC_p\subset \mcC_p'\subseteq \hat{\mcC}_p$. Aiming for a computationally tractable solution, we assume $\hat{\mcC}_p$ is also a polytope described by the same matrix $\bC_p$ as
\begin{equation}\label{eq:Cp-hat}
\hat{\mcC}_p:=\left\{\tbp:\bC_p\tbp\leq\hbd_p\right\}.
\end{equation}
The goal is to determine $\hbd_p$ so that $\mcC_p'\subseteq\hat{\mcC}_p$. A sufficient condition for this to happen is that
\begin{itemize}
    \item[\emph{c1)}] $\tbp\in\hat{\mcC}_p$ for all $\tbp\in\mcC_p$; and
    \item[\emph{c2)}] $(\bI_N+\frac{1}{1-\beta_\ell}\tbG_\ell)^\top\tbp\in\hat{\mcC}_p$ for all $\tbp\in\mcC_p$ and for all $\ell$.
\end{itemize}

For condition \emph{c1)} to be met, we must ensure that $\hbd_p\geq \bd_p$. Regarding \emph{c2)}, we first compute the following value for each constraint $i\in\{1,\ldots,I_p\}$ defining the polytope $\mcC_p$:
\begin{equation}\label{eq:gammapi}
\gamma_{p,i}=\max_\ell\max_{\bC_p\tbp\leq \bd_p}\be_i^\top \bC_p\left(\bI_N+\frac{1}{1-\beta_\ell}\tbG_\ell\right)^\top\tbp.
\end{equation}
To meet \emph{c2)}, we must select the $i$-th entry of $\hbp_P$ as $\hat{d}_{p,i}\geq \gamma_{p,i}$. To meet both \emph{c1)} and \emph{c2)}, we can set
\[\hat{d}_{p,i}=\max\{d_{p,i},\gamma_{p,i}\},\]
where $\tilde{d}_{p,i}$ is the $i$-th entry of $\bd_p$. Therefore, computing $\hbd_p$ entails solving the $I_p\times L$ linear programs (LPs) in \eqref{eq:gammapi}. The set $\hat{\mcC}_p$ is an outer approximation of $\mcC_p'$ because the latter may not necessarily include the entire line segment between $(\bI_N+\frac{1}{1-\beta_\ell}\tbG_\ell)^\top\tbp$ and $(\bI_N+\frac{1}{1-\beta_{\ell'}}\tbG_{\ell'})^\top\tbp$ for all line pairs $(\ell,\ell')$.

A similar analysis can be carried out for reactive power injections to obtain $\mcC_q'$ from $\mcC_q$ to obtain the ensuing result.

\begin{lemma}\label{le:Cq'}
According to the $Q{-}V$ model, if the training set $\mcC_q'$ includes the points 
\[\mu\bq_\mcP+(1-\mu)\left(
\left(\bI+\frac{\bG_\ell^\top}{1-\lambda_\ell}\right)\bq_{\mcP}
+\left(\frac{\ba_{\ell\mcP}\ba_{\ell\mcV}^{\top}}{x_\ell}-\bG_\ell^\top\bB_{\mcP\mcV}\right)\bu_{\mcV}\right)
\]
for all $\mu\in[0,1]$, $\bq_\mcP\in\mcC_q$, and $\ell$, then $\mcC_q\subset \mcC_q'$.
\end{lemma}

Similar to the $P{-}\theta$ model, we can use Lemma~\ref{le:Cq'} to derive an outer polytopic approximation $\hat{\mcC}_q$ of $\mcC_q'$ if the original specification set $\mcC_q$ admits the polytopic form 
\begin{equation}\label{eq:Cq-polytope}
\mcC_q:=\left\{\bq_{\mcP}:\bC_q\bq_{\mcP}\leq\bd_q\right\}.
\end{equation}
More specifically, we postulate that 
\begin{equation}\label{eq:Cq-hat}
\hat{\mcC}_q:=\left\{\bq_\mcP:\bC_q\bq_\mcP\leq\hbd_q\right\},
\end{equation}
where $\bC_q\in\mathbb{R}^{I_q\times N}$. We then compute the following value for each constraint $i\in\{1,\ldots,I_q\}$ defining $\mcC_q$:
\begin{align*}
\gamma_{q,i}
:=\max_{\ell}\max_{\substack{\bC_q\bq_{\mcP}\leq\bd_q\\
\bu_{\mcV}\in\mcC_u}}&
\be_i^\top\bC_q\left(\bI+\frac{\bG_\ell^\top}{1-\lambda_\ell}\right)\bq_{\mcP}\\
&+
\be_i^\top\bC_q\left(\frac{\ba_{\ell\mcP}\ba_{\ell\mcV}^{\top}}{x_\ell}-\bG_\ell^\top\bB_{\mcP\mcV}\right)\bu_{\mcV}
\end{align*}
Finally, we can set the $i$-th entry of $\hbd_q$ as
\[\hat{d}_{q,i}=\max\{d_{q,i},\gamma_{q,i}\}\]
for all $i\in\{1,\ldots,I_q\}$. This ensures that $\mcC_q'\subseteq \hat{\mcC}_q$, so we can safely use $\hat{\mcC}_q$ as the expanded set for training the ML model. If $\mcC_u$ is also a polytope, then finding all $\gamma_{q,i}$'s entails solving $I_q\times L$ LPs. For example, if $\mcC_u$ stems from lower and upper limits on voltage magnitudes at PV buses, then $\mcC_u$ takes the form of a hyperrectangle. 

In a nutshell, we have provided a practical method to compute the expanded training set $\mcC'$. If the original sets $\mcC_p$ and $\mcC_q$ are polytopes with $I_p$ and $I_q$ constraints, respectively, the operator can determine the sets $\hat{\mcC}_p$ and $\hat{\mcC}_q$ by solving $(I_p+I_q)L$ linear programs. Although $\hat{\mcC}$ is an outer approximation and may be conservative, our numerical tests show that an ML model with only $\mcO(N^2)$ parameters can still be trained accurately over this expanded region. This provides an approximate answer for question \emph{q2)}. We next focus on \emph{q3)}.

\section{Prediction Accuracy of Fixed-Point Solution}
\label{sec:accuracy}
Thus far, we have characterized the convergence of the fixed-point iteration and identified the specification set over which the ML model should be trained. The analysis has presumed that the ML model $F_\bw(\bc)$ approximates well the basecase AC-PF mapping $G_0(\bc)$, so that the ML-based fixed-point iterates of \eqref{eq:dnn-fixed-point} remain close to the AC-PF-based fixed-point iterates of \eqref{eq:pf-fixed-point}. This section establishes that the two trajectories indeed remain close to each other. 

For this analysis, we must quantify the sensitivity of the ML mapping $F_\bw$ to voltage phasors. This sensitivity is formally captured by the Lipschitz constant $L_\ell$ as
\begin{equation}\label{eq:dnn-fixed-point-lipschitz}
\left\|F_\bw\left(\bc+\bg_\ell(\bv)\right)-F_\bw\left(\bc+\bg_\ell(\bv')\right)\right\|_2\leq
L_\ell\|\bv-\bv'\|_2
\end{equation}
for all $\bv$ and $\bv'$ in a voltage region $\mcR$. The region $\mcR$ is defined so that $\bc+\bg_\ell(\bv)$ and $\bc+\bg_\ell(\bv')$ belong to $\mcC'$ for all $\bc\in\mcC$. The constant $L_\ell$ relates to the Jacobian of the ML model
\begin{equation}\label{eq:lipschitz}
L_\ell=\sup_{\bv\in\mcR} \left\|\nabla_{\bv}F_\bw \left(\bc+\bg_\ell(\bv)\right) \right\|_2.
\end{equation}
Our analysis requires that the ML mapping be contracting in $\bv$ as stated next.

\begin{assumption}\label{as:lipschitz}
The ML mapping is a contraction in terms of $\bv$, so that the Lipschitz constant in \eqref{eq:dnn-fixed-point-lipschitz} satisfies $L_\ell<1$ for every $\bc\in\mcC$ and every contingency $\ell$.
\end{assumption}


In~\cite{HICSS2027}, we proposed an SDP formulation for verifying that a DNN trained on the basecase topology satisfies Assumption~\ref{as:lipschitz}. We will further need to quantify the approximation error between the ML model and the AC-PF mapping as formalized in the ensuing assumption. 

\begin{assumption}\label{as:approximationerror}
Suppose the ML model approximates the basecase AC-PF mapping over the set $\mcC'$ with accuracy $\epsilon>0$
\begin{equation}\label{eq:dnn-accuracy}
\|F_\bw(\bc)-G_0(\bc)\|_2\leq\epsilon,\quad\forall \bc\in\mcC'.
\end{equation}
\end{assumption}

We can now bound the error between the ML-based and the AC-PF-based fixed-point trajectories, reproduced here as
\begin{subequations}\label{eq:redofp}
\begin{align}
\bv_\ell^{t+1}&=G_0\left(\bc+\bg_\ell(\bv_\ell^t)\right),\label{eq:redofp:acpf}\\
\hbv_\ell^{t+1}&=F_\bw\left(\bc+\bg_\ell(\hbv_\ell^t)\right).\label{eq:redofp:dnn}
\end{align}
\end{subequations}

\begin{lemma}\label{le:dnn-fixed-point-error}
Under Assumptions \ref{as:lipschitz}--\ref{as:approximationerror}, the Euclidean distance 
\begin{equation}\label{eq:iteration-error}
e_\ell^t:=\|\hbv_\ell^t-\bv_\ell^t\|_2
\end{equation}
between the fixed-point iterates in \eqref{eq:redofp} satisfies 
\begin{equation}\label{eq:error-bound}
e_\ell^t\leq \left(\frac{1-L_\ell^t}{1-L_\ell}\right)\epsilon\leq \frac{\epsilon}{1-L_\ell}
\end{equation}
if the iterates are initialized at the same point $\bv_\ell^0=\hbv_\ell^0$.
\end{lemma}

Lemma~\ref{le:dnn-fixed-point-error} establishes that although the distance between the two trajectories increases with $t$, it remains bounded. If the AC-PF fixed-point iteration converges to the true post-contingency state, the ML-based fixed-point iteration converges within a neighborhood of that state. According to \eqref{eq:error-bound}, the radius of the neighborhood depends on the prediction accuracy $\epsilon$ of the ML model and its Lipschitz constant $L_\ell$, which varies across contingencies. Because $L_\ell$ may be conservative, the resulting error bound can overestimate the actual distance between the two solutions. The numerical tests of the next section show that the solution of \eqref{eq:dnn-fixed-point} remains close to the post-contingency state across almost all contingencies.

\section{Numerical Tests}\label{sec:tests}
We evaluated the proposed methodology on the 6,717-bus Texas7k benchmark system. To learn the basecase AC-PF mapping, we trained a fully connected DNN with two hidden layers, each containing 6,717 neurons, with ReLU activation. We trained the DNN using the Adam optimizer with a learning rate of $4\times10^{-5}$, and terminated training was the loss gradient norm fell below $10^{-4}$. We conducted all numerical tests in the same RunPod environment on a compute node equipped with an NVIDIA B200 GPU and an Intel Xeon 6960P CPU. 

The Texas7k dataset provides one year of operating scenarios at a 5-min resolution. Using $k$-medoids clustering, we selected 2,000 representative scenarios. We selected 200 of these to define the specification set $\mcC$. We constructed the corresponding polytope using voltage, power injection, and line flow limits from these 200 scenarios. Following the procedure described in Sec.~\ref{sec:region}, we constructed the set $\hat{\mcC}$ to represent the training region. We generated 120,000 scenarios from $\hat{\mcC}$ and combined them with the remaining 1,800 basecase scenarios to train the DNN on both basecase operating conditions and the specifications encountered during the iterations of \eqref{eq:dnn-fixed-point}.

\begin{figure}[t]
\centering
\includegraphics[width=0.9\linewidth]
{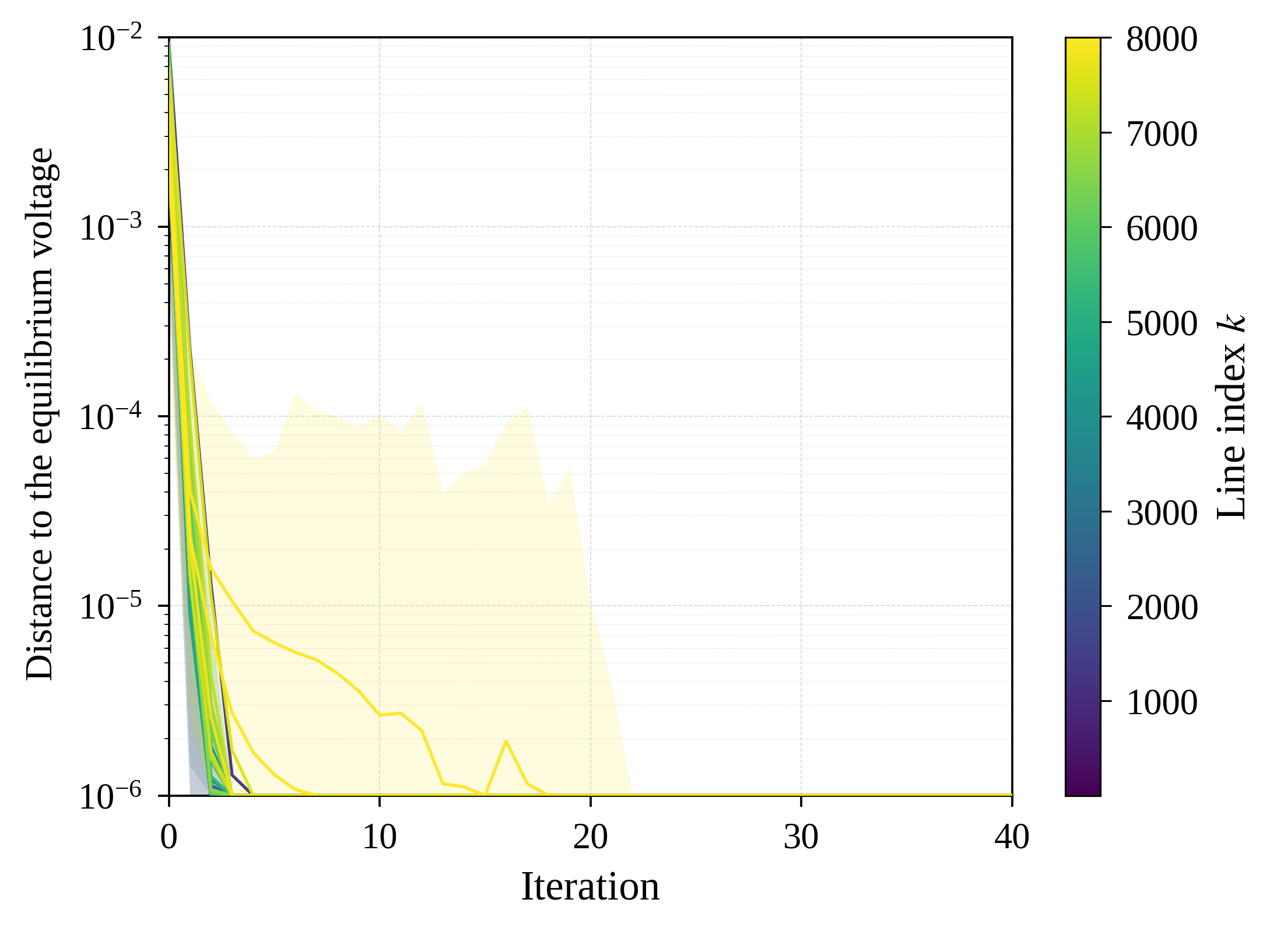}
\vspace*{-1em}
\caption{Convergence of the DNN fixed-point iteration to its equilibrium with Anderson acceleration across different contingencies. The contingencies were ordered according to empirical convergence rate and subsampled every 300.}
\label{fig:texas_fp_convergence}
\end{figure}

We first examined the convergence of \eqref{eq:dnn-fixed-point}. For every scenario, the iteration was initialized at the basecase DNN solution, and convergence was declared when the successive voltage difference norm fell below $10^{-5}$. To accelerate the fixed-point iteration, we applied Anderson acceleration (AA) with memory 3; see~\cite{walker2011anderson}. We ordered line contingencies by their iteration count to convergence, averaged across scenarios. For cleaner visualization, we sampled the ranked contingencies in increments of 300 and included the last (slowest) one. Figure~\ref{fig:texas_fp_convergence} shows the Euclidean distance between the predicted and the true post-contingency AC states across iterations. All contingencies converged within 20 iterations. 

\begin{figure}[t]
\centering
\includegraphics[width=0.95\linewidth]
{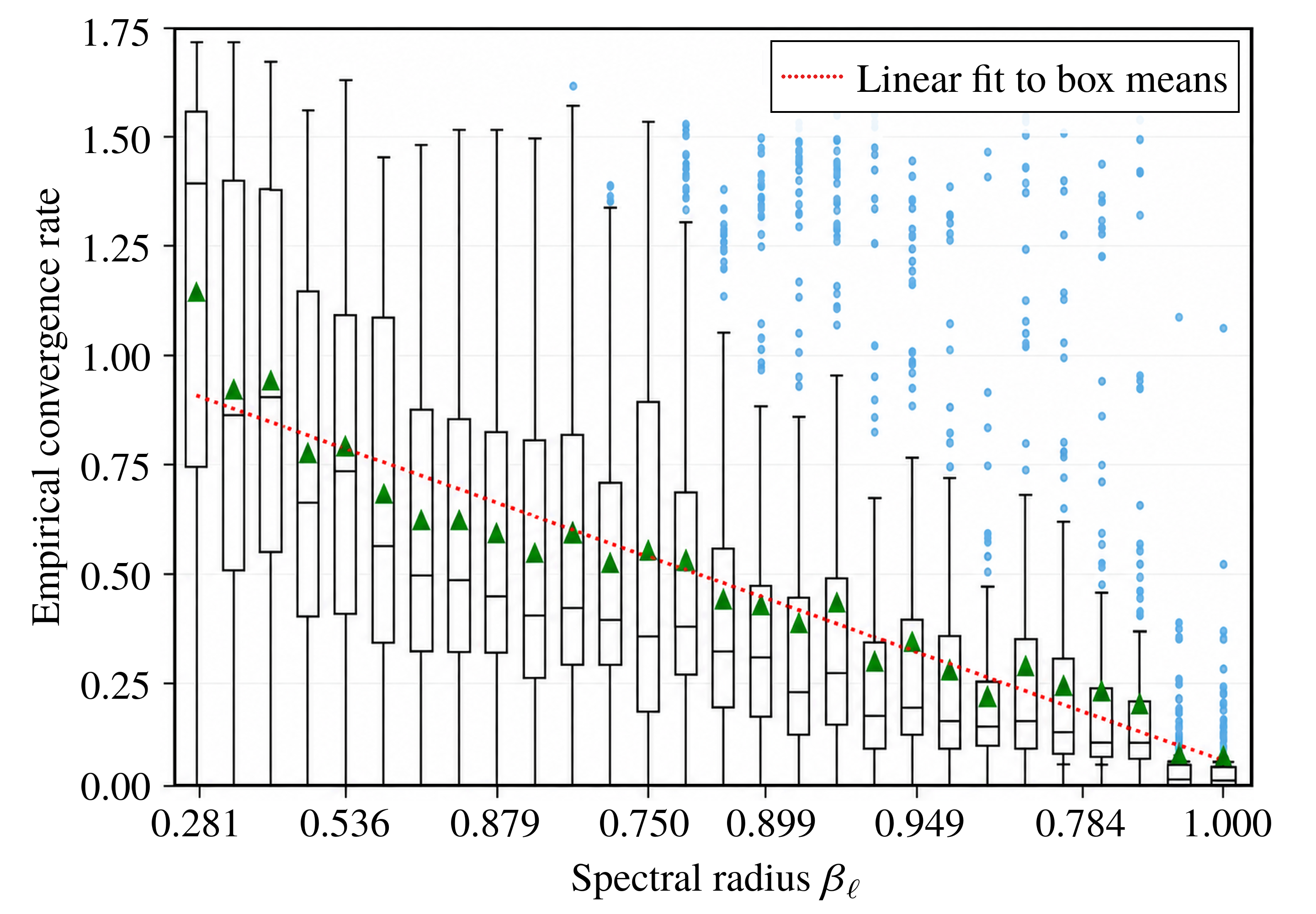}
\vspace*{-1em}
\caption{Distribution of empirical convergence rate over testing scenarios across subsampled contingencies, plotted against the corresponding spectral radii $\beta_\ell$.}
\label{fig:texas_ptheta_rate}
\end{figure}

In the second test, we explored the relationship between the empirical convergence rate and the spectral matrix radius of the $P{-}\theta$ model. We ordered contingencies by $\beta_\ell$ and sampled them again in increments of 300. The empirical convergence rate was obtained from the negative slope of the fitted line to the convergence curve in Fig.~\ref{fig:texas_fp_convergence}. Figure~\ref{fig:texas_ptheta_rate} shows the resulting distribution of empirical rates. In general, the fitted trend decreases as the radius increases, which agrees with Th.~\ref{th:pthetacontraction}. 



\begin{figure}[t]
\centering
\includegraphics[width=.95\linewidth]{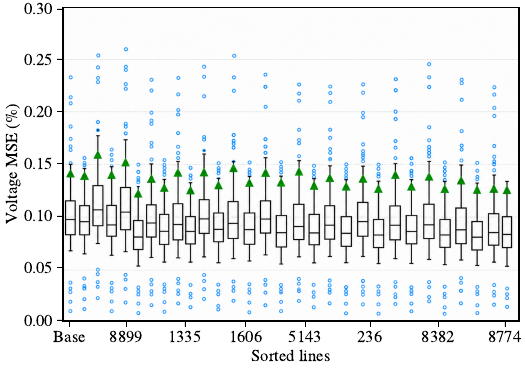}
\includegraphics[width=.95\linewidth]{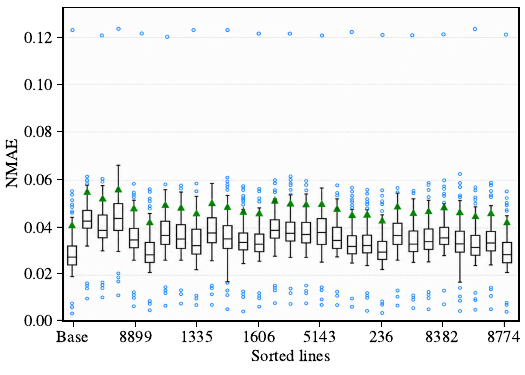}
\vspace*{-1em}
\caption{Box plots of the MSE on predicted voltages (top) and NMAE on predicted AC-PF specifications (bottom) across tested scenarios, for the basecase and 29 contingencies. The method features similar errors across the basecase and the selected contingencies.}
\label{fig:texas_mse_nmae}
\end{figure}

To assess prediction accuracy, the top panel of Fig.~\ref{fig:texas_mse_nmae} shows the MSE between the true and predicted post-contingency states across scenarios for the basecase and 29 contingencies, ordered by $\beta_\ell$ and subsampled every 300. The bottom panel shows the NMAE of the predicted AC-PF specifications, obtained by substituting the predicted states into the AC-PF equations. The MSE remains below 0.3\% and NMAE below 6\% across contingencies, comparable to the basecase. 


\begin{figure}[t]
\centering
\includegraphics[width=.95\linewidth]{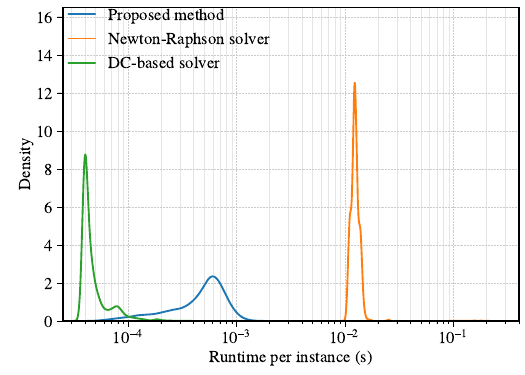}
\includegraphics[width=.95\linewidth]{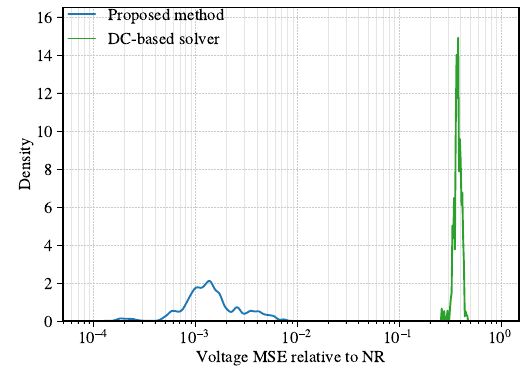}
\vspace*{-1em}
\caption{The proposed method offers a favorable trade-off between runtime (top) and accuracy (bottom) when compared to the DC solver against the NR solution. Histograms are computed across scenarios and Texas7k contingencies on the  system.}
\label{fig:texas_runtime+accuracy}
\end{figure}



Finally, we compared the compute performance of the proposed approach with NR and the DC-based solver over all testing scenarios and non-critical contingencies, totaling 1.6\,M instances. The DNN fixed-point iteration with Anderson acceleration was evaluated on an NVIDIA B200 GPU, while NR and the DC-based solver were run on the CPU, with NR contingencies parallelized across 8 workers. 
The top panel of Fig.~\ref{fig:texas_runtime+accuracy} compares the per-instance runtime distributions. The proposed method is substantially faster than NR, while the DC-based solver remains fastest. However, this speed comes at the expense of accuracy. The bottom panel of Fig.~\ref{fig:texas_runtime+accuracy} compares the voltage MSE distributions relative to the NR solution, showing substantially lower errors for the proposed method than for the DC-based solver. The proposed method offers a favorable tradeoff between the speed of the DC model and NR's accuracy.

\section{Conclusions and Future Work}\label{sec:conclusions}
We have put forth a fixed-point formulation for AC contingency analysis using a single ML model trained solely on basecase data. We established convergence guarantees and characterized convergence rates under the DC-PF approximation. We also derived an expanded training region to ensure the DNN accurately predicts AC-PF throughout the fixed-point iterations. We further established bounds on the resulting prediction errors. Numerical tests on the Texas7k system demonstrate convergence across all non-critical line outages, with prediction errors remaining close to the basecase error. An interesting direction for future work is to extend the formulation to $N-k$ contingencies. 

\appendix


\begin{IEEEproof}[Proof of Theorem~\ref{th:ptdf}]
Consider a unit power transfer from the source bus $i$ to the sink bus $j$, so that $\bp=\be_i-\be_j$ in \eqref{eq:ptheta}. By definition, the power flow on any line $\ell\in\mcL_0$ is
$\ptdf_{\ell,i,j}$. 

We will first show that $\theta_i$ is the largest voltage angle and $\theta_j$ is the smallest voltage angle across all buses. Recall that the power injection into any bus $k\in\mcN$ satisfies [cf.~\eqref{eq:ptheta}]
\begin{equation}\label{eq:pk}
p_k=\sum_{o:(k,o)\in\mcL_0}\frac{\theta_{k}-\theta_{o}}{x_{ko}}.
\end{equation}
Because $p_k=0$ for all $k\neq i,j$, it follows that $\theta_k$ cannot be the largest (smallest) angle because otherwise all numerators in the sum of \eqref{eq:pk} would be positive (negative). Therefore, it is either $\theta_i>\theta_k>\theta_j$ or $\theta_j>\theta_k>\theta_i$ for all $k\neq i,j$. Only the first option can explain $p_i=1$ and $p_j=-1$ in \eqref{eq:pk}.

Consider now the terminal buses of line $\ell=(m,n)$ and suppose $\theta_m>\theta_n$. Select a voltage angle value $\theta_0\in(\theta_n,\theta_m)$, and define the set of buses $\mcS:=\left\{q\in\mcN:\theta_q\geq\theta_0\right\}$. By construction, bus $m$ belongs to $\mcS$, and bus $n$ belongs to its complement set $\bmcS$. It also holds that $i\in\mcS$ and $j\in\bmcS$. 

Using \eqref{eq:pk}, sum the power injections into the buses in $\mcS$
\begin{equation*}
\sum_{k\in\mcS}p_k=
\sum_{k\in\mcS}\sum_{o\in\mcS:(k,o)\in\mcL_0}\frac{\theta_k-\theta_o}{x_{ko}}+\sum_{k\in\mcS}\sum_{r\in\bmcS:(k,r)\in\mcL_0}\frac{\theta_k-\theta_r}{x_{kr}}.
\end{equation*}
Each $p_k$ is expressed as a sum of the line flows leaving bus $k$. Flows are grouped into those terminating at a bus within $\mcS$, and those terminating at a bus within $\bmcS$. The first group of flows internal to $\mcS$ has zero sum. This is because every flow $(\theta_k-\theta_o)/x_{ko}$ is canceled by its reverse flow $(\theta_o-\theta_k)/x_{ko}$. Hence,
\begin{equation*}
\sum_{\substack{(k,r)\in\mcL_0\\ k\in\mcS,r\in\bmcS}}\frac{\theta_k-\theta_r}{x_{kr}}=1.
\end{equation*}
By the definition of $\mcS$, it holds $\theta_k>\theta_r$ for all $k\in\mcS$ and $r\in\bmcS$, so that all terms in the previous summand are nonnegative. Given they sum up to one, all terms must be smaller than one. The line of interest $\ell$ belongs to this cut, and so $f_\ell=(\theta_m-\theta_n)/x_{mn}\in(0,1]$.

Thus far, we assumed that $\theta_m>\theta_n$. If $\theta_m<\theta_n$, then it follows similarly that $i,n\in\mcS$ and $j,m\in\bmcS$. In this case, it follows that $f_\ell\in[-1,0)$. Combining the two cases proves that $\left|\ptdf_{\ell,i,j}\right|\leq 1$. The claim holds even if $i=m$ and $j=n$ for the self-PTDF. If line $\ell$ is the only line crossing the cut between $\mcS$ and $\bmcS$, then $\left|\ptdf_{\ell,i,j}\right|= 1$ and $\ell$ constitutes a critical line. 
\end{IEEEproof}

\begin{IEEEproof}[Proof of Theorem~\ref{th:pthetacontraction}]
The matrix $\tbG_\ell$ has rank one. Hence, its unique nonzero eigenvalue equals the trace of $\tbG_\ell$:
\begin{align*}
\tilde{\lambda}_\ell&=\frac{\tba_\ell^\top\tbB_0^{-1}\tba_\ell}{x_\ell}=\frac{\tba_\ell^\top\tbB_0^{-1}\be_m}{x_\ell}
-
\frac{\tba_\ell^\top\tbB_0^{-1}\be_n}{x_\ell}\\
&=S_{\ell,m}-S_{\ell,n}=\ptdf_{\ell,m,n}=\beta_\ell.
\end{align*}
The first equality yields $\tilde{\lambda}_\ell>0$ because $\tbB_0$ is a positive definite matrix. The second equality follows from the definition of $\tba_\ell$ in \eqref{eq:al} together with the convention for the reference bus, whereby $S_{\ell,m}$ or $S_{\ell,n}$ is set to zero whenever bus $m$ or $n$ is the reference bus, respectively. The third equality follows from the definition of $\bS$ in \eqref{eq:f=Sp}. The fourth equality follows from the PTDF definition in \eqref{eq:ptdf}. By Theorem~\ref{th:ptdf}, we get 
\[\tilde{\lambda}_\ell=\ptdf_{\ell,m,n}=\beta_\ell<1.\]

Subtracting \eqref{eq:dcpf-equilibrium} from \eqref{eq:dcpf-fixed-point} provides
\[
\btheta_\ell^{t+1}-\btheta_\ell
=
\tbG_\ell(\btheta_\ell^t-\btheta_\ell).
\]
Iterating this recursion yields
\[\btheta_\ell^{t+1}-\btheta_\ell
=
\tbG_\ell^t(\btheta_\ell^0-\btheta_\ell)
=
\beta_\ell^t\,
\tbG_\ell(\btheta_\ell^0-\btheta_\ell),
\]
where the second equality follows from the rank-one property of $\tbG_\ell$. Taking Euclidean norms on both sides gives
\[\|\btheta_\ell^t-\btheta_\ell\|_2 \le \beta_\ell^{t-1} \|\tbG_\ell(\btheta_\ell^0-\btheta_\ell)\|_2.\]
Hence, the fixed-point iteration in \eqref{eq:dcpf-fixed-point} converges to $\btheta_\ell$ because $\beta_\ell<1$, provided that $\tbG_\ell(\btheta_\ell^0-\btheta_\ell)\neq\bzero$. By the definition of $\tbG_\ell$ in \eqref{eq:dcpf-equilibrium}, the condition $\tbG_\ell(\btheta_\ell^0-\btheta_\ell)\neq\bzero$ is equivalent to $\tba_\ell^\top(\btheta_\ell^0-\btheta_\ell)\neq0$.
\end{IEEEproof}

\begin{IEEEproof}[Proof of Theorem~\ref{th:qvcontraction}]
We study three cases, depending on the types of terminal buses on the contingent line $\ell=(m,n)$. 

\emph{Case 1: Both terminal buses are PV $(m,n\in\mcV)$.} Because $\ba_{\ell_{\mcP}}=\bzero$, we get $\bG_\ell=\bzero$, $\lambda_\ell=0$, and hence, the fixed-point iteration in \eqref{eq:dcqv-fixed-point} reaches the solution after a single update.

\emph{Case 2: Both terminal buses are PQ $(m,n\in\mcP)$.} Then
\begin{equation*}
\ba_{\ell_{\mcP}}=\be_m-\be_n.
\end{equation*}
To relate the quantity in \eqref{eq:lambda-qv} to the self-PTDF, partition the reduced bus susceptance matrix into PQ and PV buses:
\begin{equation*}
\tbB_0
=
\begin{bmatrix}
\bB_{\mcP\mcP} & \tbB_{\mcP\mcV}\\
\tbB_{\mcV\mcP} & \tbB_{\mcV\mcV}
\end{bmatrix}.
\end{equation*}
The reference bus has been excluded here. Let its inverse be partitioned conformably as
\begin{equation*}
\tbC_0:=\tbB_0^{-1}=
\begin{bmatrix}
\bC_{\mcP\mcP} & \tbC_{\mcP\mcV}\\
\tbC_{\mcV\mcP} & \tbC_{\mcV\mcV}
\end{bmatrix}.
\end{equation*}
The block-inverse formula provides that
\begin{equation}\label{eq:Cpp-schur}
\bC_{\mcP\mcP}=\left(\bB_{\mcP\mcP}-\tbB_{\mcP\mcV}
\tbB_{\mcV\mcV}^{-1}\tbB_{\mcV\mcP}\right)^{-1}.
\end{equation}
Since $\tbB_0\succ\bzero$, its Schur complement satisfies
\begin{equation*}
\bzero\prec\bB_{\mcP\mcP}-\tbB_{\mcP\mcV}
\tbB_{\mcV\mcV}^{-1}\tbB_{\mcV\mcP}\preceq\bB_{\mcP\mcP}.
\end{equation*}
Inverting the two matrices yields $\bB_{\mcP\mcP}^{-1}\preceq\bC_{\mcP\mcP}$, and so $\ba_{\ell_{\mcP}}^\top\bB_{\mcP\mcP}^{-1}\ba_{\ell_{\mcP}}\leq \ba_{\ell_{\mcP}}^\top\bC_{\mcP\mcP}
\ba_{\ell_{\mcP}}$. Therefore, we get that
\begin{align*}
\lambda_\ell&=\frac{1}{x_\ell}\ba_{\ell_{\mcP}}^\top\bB_{\mcP\mcP}^{-1}
\ba_{\ell_{\mcP}}\leq \frac{1}{x_\ell}\ba_{\ell_{\mcP}}^\top\bC_{\mcP\mcP}\ba_{\ell_{\mcP}}\\
&=\frac{1}{x_\ell}\ba_{\ell}^\top\tbC_0\ba_{\ell}=\frac{1}{x_\ell}\tba_{\ell}^\top\tbB_0^{-1}\tba_{\ell}=\ptdf_{\ell,m,n}<1.
\end{align*}
The second equality holds because all the non-zero entries of $\tba_\ell$ are included in the subvector $\ba_{\ell_{\mcP}}$. Because matrix $\bB_{\mcP\mcP}$ is positive semidefinite, we also get that $\lambda_\ell\geq 0$. 


\emph{Case 3: One terminal bus is PQ and the other is PV $(m\in\mcP,n\in\mcV)$.} In this case, we get that $\ba_{\ell_{\mcP}}=\pm\be_m$ and
\begin{equation*}
\lambda_\ell=\frac{1}{x_\ell}\be_m^\top\bB_{\mcP\mcP}^{-1}\be_m.
\end{equation*}
The matrix $\bB_{\mcP\mcP}$ can be decomposed as
\begin{equation}\label{eq:Bpp-decomposition}
\bB_{\mcP\mcP}
=\frac{1}{x_\ell}
\be_m\be_m^\top+
\sum_{k\in\mcL_0\setminus\{\ell\}}
\frac{1}{x_k}
\ba_{k_{\mcP}}
\ba_{k_{\mcP}}^\top.
\end{equation}
Because line $\ell$ is non-critical, the power system remains connected after its removal. Therefore, the second term on the right-hand side of \eqref{eq:Bpp-decomposition} is a positive definite matrix, and so
\begin{equation*}
\frac{1}{x_\ell}
(\be_m^\top\bz)^2<\bz^\top
\bB_{\mcP\mcP}
\bz
\end{equation*}
for any $\bz\neq\bzero$. 
According to the variational identify for a positive definite matrix, we have that
\begin{equation*}
\be_m^\top
\bB_{\mcP\mcP}^{-1}
\be_m
=
\max_{\bz\neq\bzero}
\frac{
(\be_m^\top\bz)^2
}{
\bz^\top\bB_{\mcP\mcP}\bz
}<
x_\ell,
\end{equation*}
which implies that $\lambda_\ell\in(0,1)$.

We have established that for all three cases of terminal bus types, the spectral radius of $\bG_\ell$ satisfies $0\leq \lambda_\ell<1$ for any non-critical line. Finally, subtracting \eqref{eq:dcqv-equilibrium} from \eqref{eq:dcqv-fixed-point} and using the rank-one identity $\bG_\ell^t=\lambda_\ell^{t-1}\bG_\ell$ yields \eqref{eq:qv-rate}. 
\end{IEEEproof}

\begin{IEEEproof}[Proof of Lemma~\ref{le:Cp'}]
According to the $P{-}\theta$ model, the modified active power injections under specification $\tbp$ and contingency $\ell$ follow the discrete-time affine system
\begin{equation}\label{eq:p-fixed-point}
\tbp_\ell^{t+1}=\tbG_\ell^\top\tbp_\ell^t+\tbp.
\end{equation}
This recursion can be derived from \eqref{eq:dcpf-fixed-point} if we define $\tbp_\ell^t:=\tbB_0\tbtheta_\ell^t$ for all $t$. The goal is to characterize the set of vectors $\tbp_\ell^t$ for all contingencies $\ell$, all iterations $t$, and every $\tbp\in\mcC_p$. To this end, unrolling \eqref{eq:p-fixed-point} yields
\begin{align}\label{eq:p-trajectory}
\tbp_\ell^t&= 
\left(\bI+\sum_{\tau=1}^{t}\tbG_\ell^\tau\right)^\top\tbp=\left(\bI+\sum_{\tau=1}^{t}\beta_\ell^{\tau-1}\tbG_\ell\right)^\top\tbp\nonumber\\
&=
\left(\bI+\frac{1-\beta_\ell^t}{1-\beta_\ell}\tbG_\ell
\right)^\top\tbp=\beta_\ell^t\tbp+\left(1-\beta_\ell^t\right)\tbM_\ell\tbp.
\end{align}
where $\tbM_\ell:=\bI_N+\frac{1}{1-\beta_\ell}\tbG_\ell^\top$. The second equality follows because $\tbG_\ell$ has a single non-zero eigenvalue $\beta_\ell\in(0,1)$, and so $\tbG_\ell^\tau
=\beta_\ell^{\tau-1}\tbG_\ell$; the third stems from the geometric series finite sum; and the fourth follows upon rearranging terms.

Equation \eqref{eq:p-trajectory} shows that $\tbp_\ell^t$ is a convex combination of $\tbp$ and $\tbM_\ell\tbp$. Hence, for every $\tbp\in\mcC_p$, the training set $\mcC_p'$ should include the line segment between $\tbp$ and $\tbM_\ell\tbp$ for all $\ell$. Lemma \ref{le:Cq'} can be proved by exactly following the same procedure.
\end{IEEEproof}

\begin{IEEEproof}[Proof of Lemma~\ref{le:dnn-fixed-point-error}] We first add and subtract the vector $F_\bw\left(\bc+\bg_\ell(\bv_\ell^t)\right)$ from $\hbv_\ell^{t+1}-\bv_\ell^{t+1}$. Applying the triangle inequality to the resultant vector norm provides
\begin{align*}
e_\ell^{t+1}&\leq \left\|
F_\bw\left(\bc+\bg_\ell(\hbv_\ell^t)\right)-
F_\bw\left(\bc+\bg_\ell(\bv_\ell^t)\right)\right\|_2\\
&\quad\quad  +
\left\|F_\bw\left(\bc+\bg_\ell(\bv_\ell^t)\right)-
G_0\left(\bc+\bg_\ell(\bv_\ell^t)\right)\right\|_2.
\end{align*}
The first term is bounded by $L_\ell e_\ell^t$ from
\eqref{eq:dnn-fixed-point-lipschitz}. The second term is bounded by $\epsilon$ from \eqref{eq:dnn-accuracy} for all $\bc+\bg_\ell(\bv)\in\mcC'$. Hence, the Euclidean distance between the two trajectories satisfies $e_\ell^{t+1}\leq L_\ell e_\ell^t + \epsilon$
for all $t$. Unfolding the iteration yields
\begin{equation*}
e_\ell^t\leq L_\ell^t e_\ell^0+\epsilon\sum_{\tau=0}^{t-1}L_\ell^\tau= L_\ell^t e_\ell^0+ \left(\frac{1-L_\ell^t}{1-L_\ell}\right)\epsilon.
\end{equation*}
If the two iterations are initialized at the same voltage point so that $e_\ell^0=0$, the claim in \eqref{eq:error-bound} follows since $L_\ell\in(0,1)$. 
\end{IEEEproof}

\balance
\bibliographystyle{IEEEtran}
\bibliography{myabbrv,power,contingency,kekatos}
\end{document}